\documentclass[preprint,review]{elsarticle}

\usepackage[utf8]{inputenc}
\usepackage{hyperref}
\usepackage{graphicx}
\usepackage[cmex10]{amsmath}

\hypersetup{
     hidelinks
}

\usepackage[caption=false,font=footnotesize]{subfig}

\usepackage{algorithmic}
\usepackage{url}

\usepackage{pifont} 
\usepackage{calrsfs} 
\usepackage{amssymb, amsthm}
\usepackage{latexsym}

\usepackage{ckpt}

\usepackage{pdftricks}
\begin{psinputs}
  \usepackage{pstricks}
  \usepackage{eepic, epic}
  \usepackage[cmex10]{amsmath}
  \usepackage{amssymb, amsthm}
  \usepackage{latexsym}
  \usepackage{pifont}  \usepackage{ckpt}
 \end{psinputs}

\newtheorem{theorem}{Theorem}[section]
\newtheorem{definition}{Definition}[section]

\begin{document}

\begin{frontmatter}

\title{A Rollback in the History of Communication-Induced Checkpointing}


\author[ic]{Islene C.~Garcia\corref{cor1}}
\ead{islene@ic.unicamp.br}
\author[dcomp]{Gustavo M.~D.~Vieira}
\ead{gdvieira@ufscar.br}
\author[ic]{Luiz E.~Buzato}
\ead{buzato@ic.unicamp.br}
\cortext[cor1]{Corresponding author}

\address[ic]{Institute of  Computing, University of  Campinas, Brazil}
\address[dcomp]{Department  of Computing  at  Sorocaba, CCGT,  Federal
  University of S\~ao Carlos}

\makeatletter
\def\ps@pprintTitle{%
 \let\@oddhead\@empty
 \let\@evenhead\@empty
 \def\@oddfoot{\centerline{\thepage}}%
 \let\@evenfoot\@oddfoot}
\makeatother

\begin{abstract}
  The  literature  on  communication-induced  checkpointing  presents  a
family of protocols that use  logical clocks to control whether forced
checkpoints must be  taken. Efficiency of these  protocols is measured
by how many forced checkpoints are needed to ensure no checkpoint will
be  useless  to the  application;  the  fewer forced  checkpoints  the
better. For many years, HMNR, also called Fully Informed (FI), was the
most  complex and  efficient  protocol of  this  family.  The  Lazy-FI
protocol applies a  lazy strategy that defers the  increase of logical
clocks, resulting in a protocol with better efficiency for distributed
systems where processes can take basic checkpoints at different rates.
Recently,  the  Fully  Informed  aNd  Efficient  (FINE)  protocol  was
proposed using the same control structures  as FI, but with a stronger
and, presumably  better, checkpoint-inducing condition.  FINE  and its
lazy  version,  called Lazy-FINE,  would  now  be the  most  efficient
checkpointing protocols  based on logical clocks.   This paper reviews
this family of protocols, proves a theorem on a condition that must be
enforced by all stronger versions of FI, and proves that both FINE and
Lazy-FINE do not  guarantee the absence of useless  checkpoints.  As a
consequence, FI and Lazy-FI can be rolled back to the position of most
efficient  protocols  of  this  family  of  index-based  checkpointing
protocols.

\end{abstract}


\end{frontmatter}

\section{Introduction}
\label{sec:introduction}

Checkpointing  is   a  widely  used   technique  that
provides fault-tolerance to distributed systems. A local checkpoint is
a state of a process that can be recovered after a crash. A consistent
global checkpoint~\cite{Chandy1985} is a set of local checkpoints that
can be used  to recover a system after a  global failure. If processes
take checkpoints at their own pace, a consistent global checkpoint may
not be  formed, and, in  the worst case,  the application may  need to
rollback to its initial state after a failure.  This is the well-known
\emph{domino  effect}~\cite{Randell1975}  caused  by the  presence  of
\emph{useless} checkpoints~\cite{Netzer1995}.

Some  checkpointing protocols  avoid  useless checkpoints  by using  a
coordinator and control messages~\cite{Chandy1985,Koo1987}. Others use
a  communication-induced  approach:  processes can  take  \emph{basic}
checkpoints autonomously,  but the protocol uses  information obtained
from  the  exchange   of  messages  among  the   processes  to  induce
\emph{forced} ones and,  thus, to eliminate the  occurrence of useless
checkpoints~\cite{Elnozahy2002,Manivannan1999}.    Checkpoint-inducing
conditions based on information stored in local variables and messages
received are used to control whether a forced checkpoint must be taken
before  delivering  the  payload  of a  message  to  the  application.
Therefore, communication-induced protocols are often compared in terms
of the  number of forced checkpoints  and the size of  the state (data
structures) maintained by each process to support the decision to take
a forced checkpoint. The fewer  the number of forced checkpoints taken
and the smaller the size of the data structures used the better.

Communication-induced index-based checkpointing  protocols implement a
variant  of   Lamport's  logical  clock~\cite{Lamport1978}   to  state
checkpoint-inducing  conditions.  Protocols  that  use this  approach,
such as~\cite{Briatico1984}, \cite{Manivannan1996}, \cite{Helary2000},
and  \cite{Tsai2005},  enforce  an  easily  observable  property  that
guarantees that checkpoints  stamped with the same clock  value form a
consistent    global    checkpoint~\cite{Helary2000}.     Furthermore,
index-based protocols have presented  better efficiency than protocols
based     on      the     tracking     of      specific     checkpoint
patterns~\cite{Alvisi1999,Vieira2006}.

Unfortunately,  no matter  what mechanism  is used  to trigger  forced
checkpoints, there  is not an  optimal checkpointing protocol  for all
checkpoint and communication  patterns~\cite{Tsai1999}. However, for a
particular  family   of  protocols,  a  stronger   (more  restrictive)
condition always produces a protocol that forces fewer checkpoints than
a  protocol based  on  a  weaker condition~\cite{Tsai1999}.   Evidence
obtained from experimental comparisons of checkpointing protocols also
indicate  that  stronger conditions  usually  lead  to more  efficient
protocols~\cite{Helary2000, Luo2009}.

For many  years, the  HMNR protocol~\cite{Helary2000}  implemented the
strongest  index-based checkpoint-inducing  condition.  This  protocol
has also  been called Fully Informed  (FI)~\cite{Tsai2005}, because it
propagates  detailed   information  about  the  causal   past  of  the
processes.   Eventually,  the  literature  began to  show  efforts  to
produce     further     optimized     versions     of     FI.      The
Lazy-FI~\cite{Tsai2007-lazy}     approach     applies     the     lazy
strategy~\cite{Vieira2001a} to  increment logical  clocks of  FI.  The
Fully Informed  aNd Efficient (FINE)  protocol~\cite{Luo2008, Luo2009}
is based on  a checkpoint condition stronger than the  one defined for
FI but  using the same control  information maintained by FI.   A lazy
version    of   this    protocol,    called    Lazy-FINE   was    also
proposed~\cite{Luo2011}.   The S-FI~\cite{simon2013scalable}  protocol
aims  to  take the  same  number  of  forced  checkpoints as  FI,  but
employing a  reduced amount of  information per message  exchanged, an
improvement   that   makes   the    protocol   more   scalable.    The
DCFI~\cite{simon2013delayed} delays non-forced checkpoints in order to
reduce the total number of checkpoints in the system.

The contributions  of the  paper are three.   Firstly, it  reviews the
FI~\cite{Helary2000}  and  Lazy-FI~\cite{Tsai2007-lazy}  checkpointing
protocols to single  out the similarities in logical  structure of the
conditions used by  them to trigger forced  checkpoints.  Secondly, it
proves a theorem that shows  that the checkpoint-inducing condition of
FI cannot  be strengthened  without respecting the  timestamping rules
that guarantee the absence of  useless checkpoints.  Thirdly, it shows
that  FINE~\cite{Luo2009} and  Lazy-FINE~\cite{Luo2011} fail  the test
established by the theorem and, as a consequence, that both algorithms
do not guarantee  the absence of useless  checkpoints.  These findings
cause a rollback in the history of communication-induced checkpointing
protocols: FI and Lazy-FI are back  as the most efficient protocols of
this family of index-based checkpointing protocols.

The rest of the paper is structured as follows.
Section~\ref{sec:fundamental} presents fundamental concepts.
Section~\ref{sec:index} addresses index-based checkpointing,
describing FI~\cite{Helary2000} and Lazy-FI~\cite{Tsai2007-lazy}.
Section~\ref{sec:attempts} presents FINE~\cite{Luo2009} and
Lazy-FINE~\cite{Luo2011}, the theorem about the correctness of FI
optimizations, and the checkpoint scenarios that show that these
protocols may lead to useless checkpoints. Finally,
Section~\ref{sec:conclusion} concludes the paper.

\section{Fundamental concepts}
\label{sec:fundamental}

This section defines the meaning of distributed computation,
checkpoint, consistent global checkpoint and the mechanisms used to
track whether checkpoints belong or not to a consistent
global checkpoint.

\subsection{Distributed computation}

A set of $n$ processes ($\proc{1}$, $\dots$, $\proc{n}$) that
communicate strictly via unicast messages forms a distributed computation. The
communication graph is complete, the channels are reliable, but the
transmission delays are unpredictable. There is no global clock or
shared memory.

Every process $\proc{i}$ starts with an event $\event{i}{1}$ and
executes a possibly infinite sequence of events
$(\event{i}{1},~\event{i}{2},\ldots)$. An internal event can only
influence the state of the process that has executed it. External
events can be the sending or the receiving of messages. Given global
time  absence,  events can  be  ordered  using  solely the  notion  of
``cause-and-effect'' enabled  by the flow of  information generated by
the occurrence of internal and external events.  Thus, causality can be
captured by  the \emph{causally  precedes} relation  over events  of a
distributed computation~\cite{Lamport1978}.

\begin{definition} [Causal precedence] Event
  $\event{i}{x}$ causally precedes $\event{j}{y}$
  ($\event{i}{x} \to \event{j}{y}$) if 
  \begin{itemize}
  \item $i = j$ and $y = x+1$, or 
  \item $\exists m: \event{i}{x} = \send(m)$ and
         $\event{j}{y} = \receive(m)$, or 
  \item $\exists \event{k}{z}: \event{i}{x} \to
    \event{k}{z} \land \event{k}{z} \to \event{j}{y}$.
\end{itemize}
\end{definition}

\subsection{Fault tolerance and checkpoints} \label{sec:ckpt}

We assume  the crash-recover fault  model, that is,  in the case  of a
failure, a process halts and looses its volatile state. During correct
execution, processes frequently save  their states to stable storage
to  make possible  the recovery  of  a process  (system) by  way of  a
rollback to  an earlier  consistent state  in the  case of  partial or
total system failure.

A checkpoint is the local state of a process that was saved on stable
storage.  Every process $\proc{i}$ has an initial checkpoint, denoted
by $\ckpt{i}{1}$, and other checkpoints can be saved along the
computation.  The $x$-th checkpoint of a process $\proc{i}$ is denoted
by $\ckpt{i}{x}$. An interval $\interv{i}{x}$ is the set of events
from $\ckpt{i}{x}$ to $\ckpt{i}{x+1}$, including $\ckpt{i}{x}$ but
excluding $\ckpt{i}{x+1}$.

Let us assume that a process $\proc{i}$ manages a logical clock
$\lc_i$ that is used to timestamp a checkpoint $\ckpt{i}{x}$ with
$\ckpt{i}{x}.t$ and a message $m$ with $m.t$. The following rules
guarantee that if $\ckpt{i}{x} \to \ckpt{j}{y}$ then $\ckpt{i}{x}.t
< \ckpt{j}{y}.t$~\cite{Helary2000}. These rules can be seen as a
specialization of the Lamport's clock~\cite{Lamport1978} that
increments $\lc_i$ only at the occurrence of checkpoints.

\begin{itemize}

\item $\proc{i}$ initializes $\lc_i$ at the beginning of the computation;

\item $\proc{i}$ increments $\lc_i$ before it saves a checkpoint
  $\ckptm$ and sets $\ckptm.t = \lc_i$.

\item when $\proc{i}$ sends a message $m$, it piggybacks $\lc_i$ on the
message (denoted $m.t$);

\item when $\proc{i}$ receives a message $m$, it sets $\lc_i$ to $\max(\lc_i,
m.t)$.

\end{itemize}

Fig.~\ref{fig:ccp} depicts a distributed computation.  Horizontal
lines represent processes, one line per process.  Time flows from left
to right. Slanted arrows represent messages. Black rectangles are
basic checkpoints.  The values of the logical
clocks---timestamps---associated with each event of interest to the
distributed computation are depicted as integers between parentheses.
A checkpoint interval is represented by a left-closed right-open line
segment, for example, $\interv{3}{1}$.

\begin{figure}[htbp]
  \centering
{\begin{picture}(230,125)(0,-25)

\put(0,85){
  \put(0,0){\large $\proc{1}$}
  \put(20,0){\vector(1,0){210}}
  \put(20,0){\drawckptijts{1}{1}{1}}
  \put(40,0){\vector(1,-2){20}}
  \put(48,-15){\drawmits{1}{1}}
  \put(90,0){\drawckptijts{1}{2}{2}}
  \put(130,0){\vector(1,-2){20}}
  \put(138,-15){\drawmits{4}{2}}
  \put(205,0){\drawckptijts{1}{3}{3}}
  }

\put(0,45){
  \put(0,0){\large $\proc{2}$}
  \put(20,0){\vector(1,0){210}}
  \put(20,0){\drawckptijts{2}{1}{1}}
  \put(65,0){\vector(1,-2){20}}
  \put(47,-15){\drawmits{2}{1}}
  \put(90,0){\drawckptijts{2}{2}{2}}
  \put(120,0){\vector(1,-2){20}}
  \put(125,-15){\drawmits{3}{2}}
  \put(175,0){\vector(1,2){20}}
  \put(185,15){\drawmits{5}{2}}
   }

\put(0,5){
  \put(0,0){\large $\proc{3}$}
  \put(20,0){\vector(1,0){210}}
  \put(20,0){\drawckptijts{3}{1}{1}}
  \put(19,-20){
    \put(0,-2){\line(0,1){4}}	
    \put(0,0){\line(1,0){80}}	
    \put(40,-10){$\interv{3}{1}$}
  }
  \put(100,0){\drawckptijts{3}{2}{2}}
  \put(160,0){\drawckptijts{3}{3}{3}}
  \put(180,0){\vector(1,2){20}}
  \put(190,15){\drawmits{6}{3}}
  }

\end{picture}}
\caption{Processes, checkpoints, and logical clocks}
\label{fig:ccp}
\end{figure}

\subsection{Consistency}

A consistent global checkpoint is formed by a set of checkpoints that are unrelated by causal precedence~\cite{Chandy1985}. In Fig.~\ref{fig:ccp}, for example, the set $\{\ckpt{1}{2},~\ckpt{2}{2},~\ckpt{3}{2}\}$ is a consistent global checkpoint. 

\begin{definition} [Consistent global checkpoint] A global checkpoint $\{\ckpt{1}{x1},~\ldots,~\ckpt{n}{xn}\}$  is consistent if \\ 
\centerline {$\forall i,j : \ckpt{i}{xi} \not\to \ckpt{j}{xj}$.}
\end{definition}


When two causally unrelated checkpoints cannot be part of the same consistent global checkpoint they must be connected by a sequence of messages called a \emph{zigzag path}~\cite{Netzer1995}.

\begin{definition} [Zigzag path] \label{def_zigzag} A sequence
  of messages $\mu = [m_{l1}, \ldots, m_{lq}]$ is a zigzag path from
  $\ckpt{i}{x}$ to $\ckpt{j}{y}$ if

  \begin{itemize}
  \item $\proc{i}$ sends $m_{l1}$ after $\ckpt{i}{x}$, and
  \item if $m_{lz}$, $1 \leq z < q$, is received by $\proc{k}$,
    then $m_{lz+1}$ is sent by $\proc{k}$ in the same or a later checkpoint
    interval,  and
  \item $m_{lq}$ is received by $\proc{j}$ before $\ckpt{j}{y}$.
  \end{itemize}
\end{definition}

The existence of a zigzag  path from $\ckpt{i}{x}$ to $\ckpt{j}{y}$ is
denoted  by $\ckpt{i}{x}  \zpath \ckpt{j}{y}$.  In Fig.~\ref{fig:ccp},
$[m_1, m_2]$ is a zigzag path from $\ckpt{1}{1}$ to $\ckpt{3}{2}$ such
that $\ckpt{1}{1}$  causally precedes $\ckpt{3}{2}$, being  an example
of a  \emph{causal zigzag path}.  The zigzag  path $[m_4, m_3]$  is an
example  of  a \emph{non-causal  zigzag  path}  from $\ckpt{1}{2}$  to
$\ckpt{3}{3}$.

\begin{definition} [Z-cycle] 
$\ckpt{i}{x} \zpath \ckpt{i}{x}$
\end{definition}

A zigzag path from a checkpoint to itself forms a \emph{Z-cycle} and
makes it impossible for this checkpoint to be part of any consistent
global checkpoint. A Z-cycle is the exact condition under which a
checkpoint becomes useless~\cite{Netzer1995}. In Fig.~\ref{fig:ccp},
$\ckpt{3}{3}$ is useless due to the Z-cycles $[m_6, m_3]$ and $[m_6,
  m_5, m_4, m_3]$.
 
\subsubsection{Z-consistent timestamping}

\begin{definition} [Z-consistent timestamping] A timestamping 
  is consistent with the existence of zigzag paths if

  \centerline{$\ckpt{i}{x} \zpath \ckpt{j}{y} \Rightarrow
       \ckpt{i}{x}.t < \ckpt{j}{y}.t$}    
\end{definition}

A Z-consistent timestamping does not admit a Z-cycle, say $C \zpath
C$, since the relationship $\ckptm.t < \ckptm.t$ is impossible using
integers as timestamps~\cite{Helary1997-vprec, Helary2000}.

Fig.~\ref{fig:ccp} does not present a Z-consistent timestamping, since
$\ckpt{3}{3} \zpath \ckpt{1}{3}$ and $\ckpt{3}{3}.t =
\ckpt{1}{3}.t$. Fig.~\ref{fig:z-cons} shows the same basic checkpoint
and communication pattern of Fig.~\ref{fig:ccp} augmented with one
forced checkpoint, represented by a black diamond. This extra
checkpoint allows the Z-consistent timestamping under the rules
presented in Section~\ref{sec:ckpt}. In the next section, forced
checkpoints induced by checkpointing protocols will guarantee the
enforcement of Z-consistent timestamping.

\begin{figure}[htbp]
  \centering
{\begin{picture}(230,85)(0,-5)

\put(0,65){
  \put(0,0){\large $\proc{1}$}
  \put(20,0){\vector(1,0){210}}
  \put(20,0){\drawckptts{1}}
  \put(40,0){\vector(1,-2){15}}
  \put(48,-15){\drawmits{1}{1}}
  \put(90,0){\drawckptts{2}}
  \put(130,0){\vector(1,-2){15}}
  \put(138,-15){\drawmits{4}{2}}
  \put(205,0){\drawckptts{3}}
  }

\put(0,35){
  \put(0,0){\large $\proc{2}$}
  \put(20,0){\vector(1,0){210}}
  \put(20,0){\drawckptijts{2}{1}{1}}
  \put(65,0){\vector(1,-2){15}}
  \put(47,-15){\drawmits{2}{1}}
  \put(90,0){\drawckptts{2}}
  \put(120,0){\vector(1,-2){15}}
  \put(125,-15){\drawmits{3}{2}}
  \put(175,0){\vector(1,2){15}}
  \put(185,15){\drawmits{5}{2}}
  \put(190,0){\drawfckpttsunder{2}}
  }

\put(0,5){
  \put(0,0){\large $\proc{3}$}
  \put(20,0){\vector(1,0){210}}
  \put(20,0){\drawckptts{1}}
  \put(100,0){\drawckptts{2}}
  \put(160,0){\drawckptts{3}}
  \put(190,0){\vector(1,2){15}}
  \put(200,15){\drawmits{6}{3}}
  }

\end{picture}}
\caption{Z-consistent timestamping}
\label{fig:z-cons}
\end{figure}

\section{Index-based Checkpointing} \label{sec:index}

This  section starts  with the  description  of a  partly-informed strategy  to
induce forced checkpoints. After  that, it presents the fully-informed
and lazy strategies describing data structures and checkpoint-inducing
conditions that can be used to reduce the number of forced checkpoints
in comparison  to the partly-informed strategy. 

\subsection{Partly-informed strategy}

The partly-informed  strategy produces a Z-consistent  timestamping by
not allowing logical  clocks to decrease along a zigzag  path.  Let us
consider zigzag  paths composed by  two messages $[m_2, m_1]$,  as the
ones  depicted in  Fig.~\ref{fig:strict}.   In both  figures $m_2$  is
received  by  $\proc{i}$  after  it  has  sent  $m_1$,  which  carries
$\proc{i}$  current  $\lc_i$.  In  Fig.~\ref{fig:stricta},  $\proc{i}$
receives  $m_2$ and  $m_2.t  =  m_1.t$, thus  $\proc{i}$  is sure  the
logical clock  has not decreased  along zigzag path $[m_2,  m_1]$.  In
this case, $\proc{i}$ is not required to take a forced checkpoint, and
this absence  of action is  represented in Fig.~\ref{fig:strict}  by a
not  operator  before  a black  diamond.   In  Fig.~\ref{fig:strictb},
$\proc{i}$ receives  $m_2$, and since  $m_2.t > m_1.t$ this  implies a
decrease  in the  logical clock  along zigzag  path $[m_2,  m_1]$.  To
avoid this,  $\proc{i}$ takes  a forced  checkpoint, represented  by a
black diamond, before delivering $m_2$ to the application.

\begin{figure}[htbp]
  \centering
  \subfloat[]{
{\begin{picture}(107,97)(-7,-7)

\put(0,75){
  \put(-7,0){\large $\proc{j}$}
  \put(7,-1){$\ldots$}
  \put(20,0){\vector(1,0){80}}
  \put(23,0){\drawckptijts{j}{y}{3}}
  \put(60,0){\vector(1,-2){17.5}}
  \put(68,-15){\drawmits{2}{3}}
  }

\put(0,40){
  \put(-7,0){\large $\proc{i}$}
  \put(7,-1){$\ldots$}
  \put(20,0){\vector(1,0){80}}
  \put(23,0){\drawckptts{3}}
  \put(33,-18){\drawmits{1}{3}}
  \put(48,0){\vector(1,-2){17.5}}
  \put(60,2){\notfckptmark}		
  }

\put(0,5){
  \put(-7,0){\large $\proc{k}$}
  \put(7,-1){$\ldots$}
  \put(20,0){\vector(1,0){80}}
  \put(23,0){\drawckptts{3}}
  \put(48,2){\notfckptmark}  
  \put(80,0){\drawckptijts{k}{z}{4}}
  }

\end{picture}} \label{fig:stricta}}
  \subfloat[]{
{\begin{picture}(107,97)(-7,-7)

\put(0,75){
  \put(-7,0){\large $\proc{j}$}
  \put(7,-1){$\ldots$}
  \put(20,0){\vector(1,0){80}}
  \put(27,0){\drawckptijts{j}{y}{4}}
  \put(60,0){\vector(1,-2){17.5}}
  \put(68,-15){\drawmits{2}{4}}
  }

\put(0,40){
  \put(-7,0){\large $\proc{i}$}
  \put(7,-1){$\ldots$}
  \put(20,0){\vector(1,0){80}}
  \put(23,0){\drawckptts{3}}
  \put(48,0){\vector(1,-2){17.5}}
  \put(33,-18){\drawmits{1}{3}}
  \put(60,0){\drawfckptijts{i}{x+1}{4}}
  }

\put(0,5){
  \put(-7,0){\large $\proc{k}$}
  \put(7,-1){$\ldots$}
  \put(20,0){\vector(1,0){80}}
  \put(23,0){\drawckptts{2}}
  \put(48,2){\notfckptmark}  
  \put(80,0){\drawckptijts{k}{z}{4}}
  }

\end{picture}} \label{fig:strictb}}
  \caption{Partly-informed strategy}
  \label{fig:strict}
\end{figure}

This approach can be implemented by the following control
structures~\cite{Helary2000}:

\begin{itemize}

\item Boolean array $\sentto_i[1\dots n]$: $\sentto_i[j]$ indicates
  whether processes $\proc{i}$ has sent a message to $\proc{j}$ in the
  current checkpointing interval.

\item Array $\minto_i[1\dots n]$: $\minto_i[j]$ indicates the
  timestamp of the first message sent in the current interval by
  $\proc{i}$ to $\proc{j}$ or $+\infty$ if no such message has been
  sent.
\end{itemize}

The  partly-informed  checkpoint-inducing  condition $\cpi$,  evaluated  by
process $\proc{i}$  when it receives  a message $m$, can  be expressed
as~\cite{Helary2000}:

\begin{center}
  $\cpi \equiv \exists k: \sentto_i[k] \land
              m.t > \minto_i[k]$
\end{center}

As a consequence  of the partly-informed strategy, $\proc{k}$ does  not need to
take a checkpoint before the reception  of $m_1$, even when $m_1.t$ is
greater  than  $\lc_k$,  as   in  Fig.~\ref{fig:strictb}.

\subsection{Fully-informed strategy}

Taking into account  a zigzag path $[m_2, m_1]$  from $\ckpt{j}{y}$ to
$\ckpt{k}{z}$, like  the ones  depicted on  Fig.~\ref{fig:strict}, the
fully-informed   strategy~\cite{Helary2000}    explores   $\proc{i}$'s
information  about   $\ckpt{k}{z}.t$  to  establish   if  Z-consistent
timestamping  is being  preserved. This  information is  propagated by
piggybacking  timestamp  vectors  and checkpoint  vectors  that  carry
causal information about $\proc{k}$ to $\proc{i}$.

\subsubsection*{Strengthening the partly-informed strategy}

Let us assume that each process $\proc{i}$ maintains and propagates an
array with information about the logical clock of all processes in the
computation.  Let   us  define  the  vector   $\clockv_i$,  such  that
$\clockv_i[i]$  is equivalent  to  $\lc_i$ and  $\clockv_i[k]$ is  the
highest  value  of  $\lc_k$  that  $\proc{i}$ knows  about  due  to  a
traditional piggybacking  mechanism. Fig.~\ref{fig:clockv} illustrates
timestamp vectors, with  the logical clocks of  processes and messages
emphasized in boldface.

\begin{figure}[htbp]
  \centering
  \subfloat[$m_3.t = m_3.\clockv \lbrack 3 \rbrack
            \land\ \ckpt{1}{2}.t < \ckpt{3}{3}.t$]
           {
{\begin{picture}(230,105)(0,-10)

\put(0,75){
  \put(0,0){\large $\proc{1}$}
  \put(20,0){\vector(1,0){210}}
  \put(20,0){\drawckptcl{$\myclock{1}$ 0 0}}
  \put(50,0){\drawckptijcl{1}{2}{$\myclock{2}$ 0 0}}
  \put(155,0){\drawcl{$\myclock{2}$ 1 2}}
  \put(185,0){\vector(1,-2){17.5}}
  \put(193,-15){\drawmicl{3}{$\myclock{2}$ 1 2}}
  }

\put(0,40){
  \put(0,0){\large $\proc{2}$}
  \put(20,0){\vector(1,0){210}}
  \put(20,0){\drawckptcl{0 $\myclock{1}$ 0}}
  \put(70,0){\vector(1,-2){17.5}}
  \put(77,-15){\drawmicl{1}{0 $\myclock{1}$ 0}}
  \put(185,2){\notfckptmark}     
  \put(205,0){\drawclunder{2 $\myclock{2}$ 2}}
  }

\put(0,5){
  \put(0,0){\large $\proc{3}$}
  \put(20,0){\vector(1,0){210}}
  \put(20,0){\drawckptcl{0 0 $\myclock{1}$}}
  \put(50,0){\drawckptcl{0 0 $\myclock{2}$}}
  \put(92,0){\drawclunder{0 1 $\myclock{2}$}}
  \put(115,0){\vector(1,2){35}}
  \put(122,20){\drawmicl{2}{0 1 $\myclock{2}$}}
  \put(200,0){\drawckptijcl{3}{3}{0 1 $\myclock{3}$}}
  }				 

\end{picture}} \label{fig:clockva}}
           
  \subfloat[$m_3.t = \clockv_2 \lbrack 3 \rbrack
             \land\ \ckpt{1}{2}.t < \ckpt{3}{3}.t$]    
           {
{\begin{picture}(230,105)(0,-10)

\put(0,75){
  \put(0,0){\large $\proc{1}$}
  \put(20,0){\vector(1,0){210}}
  \put(20,0){\drawckptcl{$\myclock{1}$ 0 0}}
  \put(50,0){\drawckptijcl{1}{2}{$\myclock{2}$ 0 0}}
  \put(185,0){\vector(1,-2){17.5}}
  \put(193,-15){\drawmicl{3}{$\myclock{2}$ 0 0}}
}

\put(0,40){
  \put(0,0){\large $\proc{2}$}
  \put(20,0){\vector(1,0){210}}
  \put(20,0){\drawckptcl{0 $\myclock{1}$ 0}}
  \put(70,0){\vector(1,-2){17.5}}
  \put(77,-15){\drawmicl{1}{0 $\myclock{1}$ 0}}
  \put(115,2){\notfckptmark}     	    
  \put(140,0){\drawcl{0 $\myclock{2}$ 2}}
  \put(185,2){\notfckptmark}     
  \put(205,0){\drawclunder{2 $\myclock{2}$ 2}}
}

\put(0,5){
  \put(0,0){\large $\proc{3}$}
  \put(20,0){\vector(1,0){210}}
  \put(20,0){\drawckptcl{0 0 $\myclock{1}$}}
  \put(50,0){\drawckptcl{0 0 $\myclock{2}$}}
  \put(92,0){\drawclunder{0 1 $\myclock{2}$}}
  \put(115,0){\vector(1,2){17.5}}
  \put(122,20){\drawmicl{2}{0 1 $\myclock{2}$}}
  \put(200,0){\drawckptijcl{3}{3}{0 1 $\myclock{3}$}}
  }				 

\end{picture}} \label{fig:clockvb}}
                    
  \subfloat[$m_3.t > \lc_2 \land\ m_3.\greater \lbrack 3 \rbrack $]
           {
{\begin{picture}(230,105)(0,-10)

\put(0,75){
  \put(0,0){\large $\proc{1}$}
  \put(20,0){\vector(1,0){210}}
  \put(20,0){\drawckptcl{$\myclock{1}$\gmark \gmark}}
  \put(50,0){\drawckptijcl{1}{1}{$\myclock{2}$\gmark \gmark}}
  \put(185,0){\vector(1,-2){17.5}}
  \put(193,-15){\drawmicl{3}{$\myclock{2}$\gmark \gmark}}
  }

\put(0,40){
  \put(0,0){\large $\proc{2}$}
  \put(20,0){\vector(1,0){210}}
  \put(20,0){\drawckptcl{\gmark$\myclock{1}$\gmark}}
  \put(70,0){\vector(1,-2){17.5}}
  \put(77,-15){\drawmicl{1}{\gmark$\myclock{1}$\gmark}}
  \put(115,2){\notfckptmark}
  \put(140,0){\drawcl{\gmark$\myclock{1}$\emark}}	  
  \put(183,0){\drawfckptijcl{2}{1}{\gmark$\myclock{2}$\gmark}}
  \put(205,0){\drawclgunder{\emark$\myclock{2}$\gmark}}
  }

\put(0,5){
  \put(0,0){\large $\proc{3}$}
  \put(20,0){\vector(1,0){210}}
  \put(20,0){\drawckptcl{\gmark \gmark$\myclock{1}$}}
  \put(92,0){\drawclgunder{\gmark \emark$\myclock{1}$}}
  \put(115,0){\vector(1,2){17.5}}
  \put(125,20){\drawmicl{2}{\gmark \emark $\myclock{1}$}}
  \put(200,0){\drawckptijcl{3}{2}{\gmark \gmark$\myclock{2}$}}
  }				 

\end{picture}} \label{fig:clockvc}}
  \caption{Strengthening the partly-informed strategy}
  \label{fig:clockv}
\end{figure}

In  Fig.~\ref{fig:clockva}, there  is a  message receive  where the
$\cpi$  condition is  true. $\proc{2}$  sends a  message $m_1$  to
$\proc{3}$ with $m_1.t = 1$ and it receives $m_3$ from $\proc{1}$ with
$m_3.t = 2$.  However, since $m_3.\clockv[3]  = 2$ it does not need to
save  a   forced  checkpoint.   In  this   case,  $\ckpt{1}{2}  \zpath
\ckpt{3}{3}$,  $\ckpt{1}{2}.t  <  \ckpt{3}{3}.t$  and  a  Z-consistent
timestamping is enforced.

In  Fig.~\ref{fig:clockvb},  when  $\proc{2}$ receives  $m_2$  from
$\proc{3}$ a  similar scenario occurs. The  partly-informed condition  is true,
since $m_1.t = 1$ and $m_2.t  = 2$. However, $m_2.\clockv[3] = 2$ and
$\proc{2}$ does not need to  save a forced checkpoint. When $\proc{2}$
receives $m_3$ from $\proc{1}$ with  $m_3.t = 2$, the partly-informed condition
is  true  again.  Nevertheless,  since  $\clockv_2[3]  = 2$  a  forced
checkpoint is not necessary.  As in the previous example, $\ckpt{1}{2}
\zpath   \ckpt{3}{3}$,   $\ckpt{1}{2}.t   <   \ckpt{3}{3}.t$   and   a
Z-consistent timestamping is enforced.

A  variation $\cfione$  of the  partly-informed checkpoint-inducing  condition,
evaluated by  process $\proc{i}$  when it receives  a message  $m$ and
that  takes into  account  the  value of  $\lc_k$  up to  $\proc{i}$'s
knowledge, can be expressed as~\cite{Helary2000}:

\begin{center}
  {$\cfione \equiv \exists k: \sentto_i[k] \land
    m.t > \minto_i[k]\ \land$ \\
$m.t > max(\clockv_i[k], m.\clockv[k])$}
\end{center}

The $\cfione$ condition can also be implemented with a reduced set of
data structures~\cite{Helary2000} to minimize the cost of piggybacking
information about the processes causal past. In the FI algorithm,
there is no need for a process to know exactly which is the clock of
another process. It is just important to know if their clocks are
synchronized or not. Thus, FI can be rewritten using $\lc$ and a
vector of booleans $\greater$ instead of the vector of integers
$\clockv$.

Each entry $\greater_i[k]$ is true if to the knowledge of $\proc{i}$
its clock is greater than the clock of $\proc{k}$ ($\greater_i[k]
\equiv \clockv_i[i] > \clockv_i[k]$). When an entry $\greater_i[k]$ is
false, $\clockv_i[i]$ must be equal to $\clockv_i[k]$; their clocks
are synchronized. Due to the update rules of the timestamps, it is not
possible that $\clockv_i[i] < \clockv_i[k]$. Fig.~\ref{fig:clockvc}
presents a distributed computation with $\greater$ vectors,  but
instead of using true and false to indicate the truthfulness/falseness
of the predicate we have used the  signs $>$ or $=$. The logical clock
of processes  and messages  are emphasized in  boldface and  they are
maintained in their positions in the vectors.

When $\proc{2}$ receives $m_2$ from $\proc{3}$ with $m_2.t = 1$ it
trivially does not need to take a forced checkpoint. When $\proc{2}$
receives $m_3$ from $\proc{1}$ with $m_3.t = 2$, since
$m_3.\greater[3]$, $\proc{2}$ takes a forced checkpoint before
delivering $m_3$. The FI condition can be stated as~\cite{Helary2000}:

\begin{center}
  $\cfione \equiv \exists k: \sentto_i[k] \land  m.\greater[k]
  \land m.t > \lc_i$
\end{center}

\subsubsection*{Breaking $[\mu, m]$ Z-cycles}  

Unfortunately, the $\cfione$ condition strengthens the partly-informed strategy
more than it should. In Fig.~\ref{fig:fully}, $\proc{2}$
receives $m_3$ from $\proc{1}$ with $m_3.t = 2$ and $m_3.\greater[3]$
indicates that $\lc_3$ has already reached 2. However, if $\proc{2}$
had not taken a checkpoint before delivering $m_3$, a Z-cycle $[m_2,
  m_3, m_1]$ would had been formed and $\ckpt{3}{2}$ could had become
a useless checkpoint.
Checkpoint vectors with taken marks can be used to prevent this
Z-cycle and any one composed by a causal component $\mu$ and a
single message $m$. 

\begin{figure}[htbp]
  \centering  
{\begin{picture}(230,130)(0,-15)

\put(0,85){
  \put(0,0){\large $\proc{1}$}
  \put(20,0){\vector(1,0){210}}
  \put(20,0){\drawckpttabcldv{$\myclock{1}$ & \gmark & \gmark}
                             {\nmark{1} & \umark{0} & \umark{0}}}
  \put(60,0){\drawckptijtabcldv{1}{2}{$\myclock{2}$ & \gmark & \gmark}
                                     {\nmark{2} & \umark{0} & \umark{0}}}
  \put(140,0){\drawtabcldv{$\myclock{2}$ & \gmark & \emark}
                          {\nmark{2} & \umark{1} & \nmark{2}}}
  \put(178,0){\vector(1,-2){20}}
  \put(196,-10){\drawmitabcldv{3}{$\myclock{2}$ & \gmark & \emark}
                          {\nmark{2} & \umark{1} & \nmark{2}}}
  }

\put(0,45){
  \put(0,0){\large $\proc{2}$}
  \put(20,0){\vector(1,0){210}}
  \put(20,0){\drawckpttabcldv{\gmark & $\myclock{1}$ & \gmark}
                             {\umark{0} & \nmark{1} & \umark{0}}}
  \put(40,0){\vector(1,-2){20}}
  \put(56,-10){\drawmitabcldv{1}{\gmark & $\myclock{1}$ & \gmark}
      {\umark{0} & \nmark{1} & \umark{0}}}
  \put(180,0){\drawfckptijtabcldv{2}{2}{\gmark & $\myclock{2}$ & \gmark}
               {\umark{0} & \nmark{2} & \umark{0}}}
  \put(205,0){\drawtabcldvunder{\emark & $\myclock{2}$ & \emark}{2 & 2 & 2}}
  }

\put(0,5){
  \put(0,0){\large $\proc{3}$}
  \put(20,0){\vector(1,0){210}}
  \put(20,0){\drawckpttabcldv{\gmark & \gmark & $\myclock{1}$}
                             {\umark{0} & \umark{0} & \nmark{1}}}
  \put(60,0){\drawtabcldvunder{\gmark & \emark & $\myclock{1}$}{\umark{0} & 1 & 1}}
  \put(90,0){\drawckptijtabcldv{3}{2}{\gmark & \gmark & $\myclock{2}$}{\umark{0} & \umark{1} & \nmark{2}}}
  \put(120,0){\vector(1,4){20}}
  \put(130,26){\drawmitabcldv{2}{\gmark & \gmark & $\myclock{2}$}
   	 {\umark{0} & \umark{1} & \nmark{2}}}
  }				 

\end{picture}}
  \caption{Fully-informed strategy}
  \label{fig:fully}
\end{figure}

\begin{figure}[htbp]
  \centering
  \begin{tabular}{|c|} \hline 
    {\parbox{9.0cm}{
        \begin{algorithmic}
          \STATE \hspace{-1em}\textit{take$\_$checkpoint():}          
          \FORALL{$k$} 
          \STATE $\sentto_i[k]$ \attr \FALSE;
          \ENDFOR
          \FORALL{$k \neq i$} 
          \STATE $\taken_i[k]$ \attr \TRUE;
          \STATE $\greater_i[k]$ \attr \TRUE;
          \ENDFOR
          \STATE $\lc_i$ \attr $\lc_i + 1$; 
          \STATE \textit{Save the current state on stable memory;}
          \STATE $\ckptv_i[i]$ \attr $\ckptv_i[i]+1$; 
        \end{algorithmic}
        \medskip
        \begin{algorithmic}
          \STATE \hspace{-1em} \textbf{$P_i$'s initialization:}          
          \FORALL{$k$} 
          \STATE $\ckptv_i[k]$ \attr 0; \ENDFOR
          \STATE $\lc_i$ \attr $0$;          
          \STATE $\taken_i[i]$ \attr \FALSE;
          \STATE $\greater_i[i]$ \attr \FALSE;
          \STATE \textit{take$\_$checkpoint()};
        \end{algorithmic}
        \medskip
        \begin{algorithmic}
          \STATE \hspace{-1em}\textbf{$P_i$ sends a message to $P_k$:}
          \STATE $\sentto_i[k]$ \attr \TRUE;
          \STATE $\send(m, \lc_i, \greater_i, \ckptv_i, \taken_i)$ to $P_k$;
        \end{algorithmic}
    }} \\ \hline
  \end{tabular}
\caption{FI protocol~\cite{Helary2000} (Part 1)} \label{alg:orig-HMNR1}
\end{figure}

\begin{figure}[htbp]
  \centering
  \begin{tabular}{|c|} \hline  
    {\parbox{9.0cm}{
        \begin{algorithmic} 
          \STATE \hspace{-1em}\textit{FI$\_1()$:}
          \RETURN {$\exists k:\sentto_i[k] \land m.\greater[k]
            \land m.t > \lc_i;$}
        \end{algorithmic} 
        \medskip
        \begin{algorithmic}
          \STATE \hspace{-1em}\textit{FI$\_2()$:}          
          \RETURN $m.\ckptv[i] = \ckptv_i[i] \land m.\taken[i]$;
        \end{algorithmic}        
        \smallskip
        \begin{algorithmic} 
          \STATE \hspace{-1em}\textbf{$P_i$ receives a message from $P_j$:}
          \IF {\textit{FI$\_1() \lor$ FI$\_2()$}}
          \STATE \textit{take$\_$checkpoint()};
          \ENDIF
          \IF {$m.t > \lc_i$}
          \STATE $\lc_i$ \attr $m.t$
          \FORALL{$k \neq i$}
          \STATE $\greater_i[k] = m.\greater[k]$;
          \ENDFOR
          \ELSIF {$m.t = \lc_i$}
          \FORALL{$k \neq i$}
          \STATE $\greater_i[k] = \greater_i[k] \land m.\greater[k]$;
          \ENDFOR
          \ENDIF
          \FORALL{$k \neq i$}          
          \IF {$m.\ckptv[k] > \ckptv_i[k]$}
          \STATE $\ckptv_i[k]$ \attr $m.\ckptv[k]$; 
          \STATE $\taken_i[k]$ \attr $m.\taken[k]$;
          \ELSIF {$m.\ckptv[k] = \ckptv_i[k]$}
          \STATE $\taken_i[k]$ \attr $\taken_i[k] \lor m.\taken[k]$;
          \ENDIF
          \ENDFOR
          \STATE \textit{$\deliver(m)$}
        \end{algorithmic}
    }} \\ \hline
  \end{tabular}
\caption{FI protocol~\cite{Helary2000}  (Part 2)} \label{alg:orig-HMNR2}
\end{figure}

Let us  assume that each  process $\proc{i}$ maintains and  propagates a
variation of the traditional vector clock~\cite{Fidge1991} that counts
how many checkpoints have been taken during the computation.  The
entry $\ckptv_i[i]$ expresses exactly  the number of checkpoints taken
by $\proc{i}$ and $\ckptv_i[k]$ counts how many checkpoints $\proc{k}$
has taken to the best knowledge of $\proc{i}$.

An extra boolean array  $\taken_i$\footnote{The $\taken$ array has the
  opposite    meaning    of     the    boolean    array    $\simple_i$
  ($\taken_i[j] \equiv \lnot \simple_i[j]$), defined in the context of
  a checkpointing  protocol~\cite{BHR2001} that enforces  the Rollback
  Dependency Trackability  property~\cite{Wang1997}.}  can be  used to
indicate  if the  causal  components ending  in  the current  interval
contain a  checkpoint. An entry  $\taken_i[k]$ is  true if there  is a
causal     zigzag     path     from     $\ckpt{k}{\ckptv_i[k]}$     to
$\ckpt{i}{\ckptv_i[i] +  1}$ and  this causal  zigzag path  includes a
checkpoint~\cite{Helary2000}.      Fig.~\ref{fig:fully}    shows     a
computation with $\greater$ vectors  and checkpoint vectors with taken
marks; an  entry $k$  of the  checkpoint array  is underlined  only if
$\taken[k]$ is true.

The condition to break $[\mu, m]$ Z-cycles can be expressed using the
following condition $\cfitwo$, evaluated by process $\proc{i}$ when it
receives a message $m$~\cite{Helary2000}:

\begin{center}
  $\cfitwo \equiv m.ckpt[i] = ckpt_i[i] \land m.\taken[i]$
\end{center}

The checkpoint-inducing condition of FI is an or-operation of the conditions
explained before:

\begin{center}
  $\cfi \equiv \cfione \lor \cfitwo$
\end{center}

Figs.~\ref{alg:orig-HMNR1} and ~\ref{alg:orig-HMNR2} present the code that implements FI~\cite{Helary2000}. 

\subsection{Lazy strategy}

The lazy strategy  reduces the number of  forced checkpoints necessary
to  produce a  Z-consistent  timestamping by  detecting  when a  basic
checkpoint can be  taken by a process  $\proc{i}$ without incrementing
its logical clock $\lc_i$.  Let us consider the situation of a process
$\proc{i}$  that  receives a  single  message  $m_1$ in  a  checkpoint
interval and later  decides to take a basic checkpoint  that ends this
interval,    as   depicted    in   Fig.~\ref{fig:lazystrategy}.     If
$m_1.t <  \ckpt{i}{x}.t$, $\proc{i}$ can  reuse the same  timestamp of
$\ckpt{i}{x}$       to       label       $\ckpt{i}{x+1}$       because
$\ckpt{k}{z}.t    <    \ckpt{i}{x+1}.t$    will   still    hold,    as
Fig.~\ref{fig:lazystrategya}      shows.      However,      if
$m_1.t \geq \ckpt{i}{x}.t$, $\proc{i}$ must increment $\lc_i$ to label
$\ckpt{i}{x+1}$ in order to  produce a Z-consistent timestamping where
$\ckpt{k}{z}.t     <     \ckpt{i}{x+1}.t$,      as     depicted     in
Fig.~\ref{fig:lazystrategyb} and Fig.~\ref{fig:lazystrategyc}.

\begin{figure}[htbp]
  \centering
  \subfloat[]{
{\begin{picture}(95,60)(-7,-7)

\put(0,40){
  \put(-7,0){\large $\proc{i}$}
  \put(7,-1){$\ldots$}
  \put(20,0){\vector(1,0){70}}
  \put(23,0){\drawckptijts{i}{x}{2}}
  \put(65,0){\drawckptijts{i}{x+1}{2}}
  }

\put(0,5){
  \put(-7,0){\large $\proc{k}$}
  \put(7,-1){$\ldots$}
  \put(20,0){\vector(1,0){70}}
  \put(23,0){\drawckptijts{k}{z}{1}}
  \put(35,0){\vector(1,2){17.5}}
  \put(42,10){\drawmits{1}{1}}
  }

\end{picture}} \label{fig:lazystrategya}}
  
  \subfloat[]{
{\begin{picture}(95,60)(-7,-7)

\put(0,40){
  \put(-7,0){\large $\proc{i}$}
  \put(7,-1){$\ldots$}
  \put(20,0){\vector(1,0){70}}
  \put(23,0){\drawckptijts{i}{x}{2}}
  \put(50,4){$+$}
  \put(65,0){\drawckptijts{i}{x+1}{3}}
  }

\put(0,5){
  \put(-7,0){\large $\proc{k}$}
  \put(7,-1){$\ldots$}
  \put(20,0){\vector(1,0){70}}
  \put(23,0){\drawckptijts{k}{z}{2}}
  \put(35,0){\vector(1,2){17.5}}
  \put(42,10){\drawmits{1}{2}}
  }

\end{picture}} \label{fig:lazystrategyb}}
  \subfloat[]{
{\begin{picture}(95,60)(-7,-7)

\put(0,40){
  \put(-7,0){\large $\proc{i}$}
  \put(7,-1){$\ldots$}
  \put(20,0){\vector(1,0){70}}
  \put(23,0){\drawckptijts{i}{x}{2}}
  \put(50,4){$+$}  
  \put(65,0){\drawckptijts{i}{x+1}{4}}
  }

\put(0,5){
  \put(-7,0){\large $\proc{k}$}
  \put(7,-1){$\ldots$}
  \put(20,0){\vector(1,0){70}}
  \put(23,0){\drawckptijts{k}{z}{3}}
  \put(35,0){\vector(1,2){17.5}}
  \put(42,10){\drawmits{1}{3}}
  }

\end{picture}} \label{fig:lazystrategyc}}
  \caption{Lazy indexing strategy}
  \label{fig:lazystrategy}
\end{figure}

The  lazy   strategy  can  be   implemented  by  introducing   a  flag
$\increment$ that signals $\proc{i}$  it must increment $\lc_i$ before
taking a basic checkpoint. This flag  is set to false in the beginning
of each checkpoint interval and is set to true whenever a message with
$m.t \geq lc_i$ is received.  The  setting of the $\increment$ flag is
illustrated in Fig.~\ref{fig:lazystrategy} by a $+$ sign.

\subsubsection*{Lazy-FI protocol}

Let us try to apply the lazy approach to FI using the vector
$\greater$. In Fig.~\ref{fig:lazy-greatera}, when $\proc{2}$
receives $m_5$ from $\proc{1}$ with a greater clock, it can verify
that $\proc{3}$ have already reached the same clock. However, due to
the lazy strategy, $\proc{2}$ does not know whether $\proc{3}$ will
increase its clock to save the next checkpoint. Thus, a forced
checkpoint before the delivering of $m_5$ will be required in order to
guarantee a Z-consistent timestamping.

\begin{figure}[htbp]
  \centering
\begin{picture}(230,130)(0,-5)

  \put(0,110){
    \put(0,0){\large $\proc{1}$}
    \put(20,0){\vector(1,0){210}}
    \put(20,0){\drawckptcl{$\myclock{1}$\gmark \gmark \gmark}}
    \put(60,4){$+$}
    \put(90,0){\drawckptijcl{1}{2}{$\myclock{2}$\gmark \gmark \gmark}}
    \put(156,4){$+$}
    \put(195,0){\vector(1,-2){17.5}}
    \put(170,-15){\drawmicl{5}{$\myclock{2}$\gmark \emark \gmark}}
}

\put(0,75){
  \put(0,0){\large $\proc{2}$}
  \put(20,0){\vector(1,0){210}}
  \put(20,0){\drawckptcl{\gmark $\myclock{1}$\gmark \gmark}}
  \put(90,0){\vector(1,-2){17.5}}
  \put(100,-18){\drawmicl{3}{\gmark $\myclock{1}$\gmark \gmark}}
  \put(200,0){\drawfckpttsunder{\emark \gmark$\myclock{2}$\gmark}}
  \put(212,4){$+$}		         
}

\put(0,40){
  \put(0,0){\large $\proc{3}$}
  \put(20,0){\vector(1,0){210}}
  \put(20,0){\drawckptcl{\gmark \gmark $\myclock{1}$\gmark}}
  \put(45,0){\vector(1,4){17.5}}
  \put(54,50){\drawmicl{1}{\gmark \gmark $\myclock{1}$\gmark}}
  \put(64,4){$+$}	
  \put(90,0){\drawckptclunder{\gmark \gmark $\myclock{2}$\gmark}}
  \put(140,0){\vector(1,4){17.5}}
  \put(116,50){\drawmicl{4}{\gmark \gmark$\myclock{2}$ \gmark}}
  \put(200,0){\drawckptijcl{3}{3}{\gmark \gmark$\myclock{2}$ \gmark}}
}

\put(0,5){
  \put(0,0){\large $\proc{4}$}
  \put(20,0){\vector(1,0){210}}
  \put(20,0){\drawckptcl{\gmark \gmark \gmark$\myclock{1}$}}
  \put(50,0){\vector(1,2){17.5}}
  \put(56,14){\drawmicl{2}{\gmark \gmark \gmark$\myclock{1}$}}
}

\end{picture}
  \caption{$\lc$ and $\greater$ are not enough to implement Lazy-FI}
  \label{fig:lazy-greatera}
\end{figure}
           
\begin{figure}[htbp]
  \centering
  \subfloat[$\lc$ and $\equalincr$]
           {
\begin{picture}(230,170)(0,-5)

  \put(0,145){
    \put(0,0){\large $\proc{1}$}
    \put(20,0){\vector(1,0){210}}
    \put(20,0){\drawckptcl{$\myclock{1}$\geqmark \geqmark \geqmark \geqmark}}
    \put(60,4){$+$}
    \put(100,0){\drawckptijcl{1}{2}{$\myclock{2}$\geqmark \geqmark \geqmark \geqmark}}    
    \put(164,4){$+$}
    \put(202,0){\vector(1,-2){17.5}}
    \put(175,-15){\drawmicl{7}{$\tsincr{\myclock{2}}$\geqmark \imark \geqmark \geqmark}}
}

\put(0,110){
  \put(0,0){\large $\proc{2}$}
  \put(20,0){\vector(1,0){210}}
  \put(20,0){\drawckptcl{$\geqmark \myclock{1} \geqmark \geqmark \geqmark$}}
  \put(80,0){\vector(1,-2){17.5}}
  \put(88,-15){\drawmicl{4}{$\geqmark \myclock{1} \geqmark \geqmark \geqmark$}}
  \put(200,2){\notfckptmark}    
}

\put(0,75){
  \put(0,0){\large $\proc{3}$}
  \put(20,0){\vector(1,0){210}}
  \put(20,0){\drawckptcl{\geqmark \geqmark $\myclock{1}$\geqmark \geqmark}}
  \put(45,0){\vector(1,4){17.5}}
  \put(58,50){\drawmicl{2}{\geqmark \geqmark $\myclock{1}$\geqmark \geqmark}} 
  \put(58,4){$+$}	
  \put(80,0){\drawckptclunder{\geqmark \geqmark $\myclock{2}$\geqmark \geqmark}}
  \put(120,4){$+$}
  \put(150,0){\vector(1,4){17.5}}
  \put(128,50){\drawmicl{6}{\geqmark\geqmark$\tsincr{\myclock{2}}$\geqmark \geqmark}} 
  \put(200,0){\drawckptijcl{3}{3}{\geqmark \geqmark$\myclock{3}$\geqmark \geqmark}}
}

\put(0,40){
  \put(0,0){\large $\proc{4}$}
  \put(20,0){\vector(1,0){210}}
  \put(20,0){\drawckptcl{\geqmark \geqmark \geqmark $\myclock{1}$\geqmark}}
  \put(36,0){\vector(1,-2){17.5}}	
  \put(46,-18){\drawmicl{1}{\geqmark \geqmark \geqmark $\myclock{1}$\geqmark}}
  \put(45,0){\vector(1,2){17.5}}
  \put(51,14){\drawmicl{3}{\geqmark \geqmark \geqmark $\myclock{1}$\geqmark}}
}

\put(0,05){
  \put(0,0){\large $\proc{5}$}
  \put(20,0){\vector(1,0){210}}
  \put(20,0){\drawckptcl{\geqmark \geqmark \geqmark \geqmark $\myclock{1}$}}
  \put(50,-8){$+$}
  \put(80,0){\drawckptcl{\geqmark \geqmark \geqmark \geqmark $\myclock{2}$}}
  \put(104,0){\vector(1,4){17.5}}
  \put(116,50){\drawmicl{5}{\geqmark \geqmark \geqmark \geqmark $\myclock{2}$}}
}

\end{picture}
 \label{fig:lazy-greaterb}}
           
  \subfloat[$\lc$, $\equalincr$, $ckptv$, and $\taken$]
           {
\begin{picture}(230,200)(0,-5)

  \put(0,165){
    \put(0,0){\large $\proc{1}$}
    \put(20,0){\vector(1,0){210}}
    \put(20,0){\drawckpttabcldv{$\myclock{1}$ & \geqmark & \geqmark & \geqmark & \geqmark}
                                 {\nmark{1} & \umark{0} & \umark{0} & \umark{0} & \umark{0}}}
    \put(62,4){$+$}
    \put(100,0){\drawckptijtabcldv{1}{2}{$\myclock{2}$ & \geqmark & \geqmark & \geqmark & \geqmark}
    {\nmark{2} & \umark{0} & \umark{1} & \umark{0} & \umark{0}}}
    \put(166,4){$+$}
    \put(186,0){\vector(1,-4){10}}
    \put(196,-10){\drawmitabcldv{7}{$\tsincr{\myclock{2}}$ & \geqmark & \imark & \geqmark & \geqmark}
                   {\nmark{2} & \umark{1} & \nmark{2} & \umark{1} & \nmark{2}}}
}

\put(0,125){
  \put(0,0){\large $\proc{2}$}
  \put(20,0){\vector(1,0){210}}
  \put(20,0){\drawckpttabcldv{\geqmark & $\myclock{1}$ & \geqmark & \geqmark & \geqmark}
               {\umark{0} & \nmark{1} & \umark{0} & \umark{0} & \umark{0}}}
  \put(70,0){\vector(1,-4){10}}
  \put(80,-12){\drawmitabcldv{4}{\geqmark & $\myclock{1}$ & \geqmark & \geqmark & \geqmark}
              {\umark{0} & \nmark{1} & \umark{0} & \umark{0} & \umark{0}}}
   \put(186,0){\drawfckpttabcldv{\geqmark & \geqmark & $\myclock{1}$ & \geqmark & \geqmark}
                      {\umark{0} & \umark{0} & \nmark{2} & \umark{0} & \nmark{0}}}
   \put(198,-8){$+$}		      
}

\put(0,85){
  \put(0,0){\large $\proc{3}$}
  \put(20,0){\vector(1,0){210}}
  \put(20,0){\drawckpttabcldv{\geqmark & \geqmark & $\myclock{1}$ & \geqmark & \geqmark}
            {\umark{0} & \umark{0} & \nmark{1} & \umark{0} & \umark{0}}}
  \put(45,0){\vector(1,4){20}}
  \put(65,67){\drawmitabcldv{2}{\geqmark & \geqmark & $\myclock{1}$ & \geqmark & \geqmark}
             {\umark{0} & \umark{0} & \nmark{1} & \umark{0} & \umark{0}}}
  \put(62,4){$+$}	
  \put(100,0){\drawckpttabcldvunder{\geqmark & \geqmark & $\myclock{2}$ & \geqmark & \geqmark}
            {\umark{0} & \umark{1} & \nmark{2} & \umark{1} & \umark{0}}}
  \put(138,4){$+$}
  \put(150,0){\vector(1,4){20}}
  \put(131,70){\drawmitabcldv{6}{\geqmark & \geqmark & $\tsincr{\myclock{2}}$ & \geqmark & \geqmark}
                        {\umark{0} & \umark{1} & \nmark{2} & \umark{1} & \nmark{2}}}
}

\put(0,45){
  \put(0,0){\large $\proc{4}$}
  \put(20,0){\vector(1,0){210}}
  \put(20,0){\drawckpttabcldv{\geqmark & \geqmark & \geqmark & $\myclock{1}$ & \geqmark}
            {\umark{0} & \umark{0} & \umark{0} & \nmark{1} & \umark{0}}}
  \put(45,0){\vector(1,-4){10}}	
  \put(55,-8){\drawmitabcldv{1}{\geqmark & \geqmark & \geqmark & $\myclock{1}$ & \geqmark}
               {\umark{0} & \umark{0} & \umark{0} & \nmark{1} & \umark{0}}}
  \put(55,0){\vector(1,4){10}}
  \put(65,28){\drawmitabcldv{3}{\geqmark & \geqmark & \geqmark & $\myclock{1}$ & \geqmark}
          {\umark{0} & \umark{0} & \umark{0} & \nmark{1} & \umark{0}}}
}

\put(0,05){
  \put(0,0){\large $\proc{5}$}
  \put(20,0){\vector(1,0){210}}
  \put(20,0){\drawckpttabcldv{\geqmark & \geqmark & \geqmark & \geqmark & $\myclock{1}$}
            {\umark{0} & \umark{0} & \umark{0} & \umark{0} & \nmark{1}}}
  \put(56,-8){$+$}
  \put(90,0){\drawckpttabcldv{\geqmark & \geqmark & \geqmark & \geqmark & $\myclock{2}$}
                      {\umark{0} & \umark{0} & \umark{0} & \umark{1} & \nmark{2}}}
  \put(120,0){\vector(1,4){20}}
  \put(140,67){\drawmitabcldv{5}{\geqmark & \geqmark & \geqmark & \geqmark & $\myclock{2}$}
                      {\umark{0} & \umark{0} & \umark{0} & \umark{1} & \nmark{2}}}
}

\end{picture}
 \label{fig:lazy-greaterc}}
           
  \caption{Timestamp information to implement Lazy-FI}
  \label{fig:lazy-greater2}
\end{figure}

\begin{figure}[htbp]
  \centering
  \begin{tabular}{|c|} \hline
    {\parbox{9.0cm}{
        \begin{algorithmic}
          \STATE \hspace{-1em}\textit{take$\_$checkpoint():}
          \FORALL{$k$} 
          \STATE $\sentto_i[k]$ \attr \FALSE; \ENDFOR
          \FORALL{$k \neq i$} 
          \STATE $\taken_i[k]$ \attr \TRUE; \ENDFOR
          \IF{$\increment$}
          \STATE $\lc_i$ \attr $\lc_i + 1$; 
          \FORALL{$k$}     
          \STATE $\equalincr_i[k]$ \attr \FALSE; \ENDFOR \ENDIF
          \STATE $\increment$ \attr \FALSE;
          \STATE \textit{Save the current state on stable memory;}
          \STATE $\ckptv_i[i]$ \attr $\ckptv_i[i]+1$; 
        \end{algorithmic}
        \medskip
        \begin{algorithmic}
          \STATE \hspace{-1em} \textbf{$P_i$'s initialization:}
          \FORALL{$k$} 
          \STATE $\ckptv_i[k]$ \attr 0; \ENDFOR
          \STATE $\lc_i$ \attr $0$;
          \STATE $\increment$ \attr \TRUE;
          \STATE $\taken_i[i]$ \attr \FALSE;
          \STATE \textit{take$\_$checkpoint()};
        \end{algorithmic}
        \medskip        
        \begin{algorithmic}
          \STATE \hspace{-1em}\textbf{$P_i$ sends a message to $P_k$:}
          \STATE $\sentto_i[k]$ \attr \TRUE;
          \STATE $\send(m, \lc_i,$
          \STATE \hspace{2em}  $\equalincr_i, \ckptv_i, \taken_i)$ to $P_k$;
        \end{algorithmic}
    }} \\ \hline
  \end{tabular}
  \caption{Lazy-FI protocol (adapted from~\cite{Tsai2007-lazy}) (Part 1)} \label{alg:lazy-fi1}
\end{figure}

\begin{figure}[htbp]
  \centering
  \begin{tabular}{|c|} \hline
    {\parbox{9.0cm}{
        \begin{algorithmic}
          \STATE \hspace{-1em}\textit{LAZY\_FI$\_1()$:}
          \RETURN $\exists k:\sentto_i[k] \land \lnot m.\equalincr[k] \land
                  m.t > \lc_i$;
        \end{algorithmic}
        \medskip

        \begin{algorithmic}
          \STATE \hspace{-1em}\textit{LAZY\_FI$\_2()$:}
          \RETURN $m.\ckptv[i] = \ckptv_i[i] \land m.taken[i]$;
        \end{algorithmic}

        \medskip
        \begin{algorithmic} 
          \STATE \hspace{-1em}\textbf{$P_i$ receives a message from $P_j$:}
          \IF {\textit{LAZY\_FI$\_1() \lor$ LAZY\_FI$\_2()$}}
          \STATE \textit{take$\_$checkpoint()};
          \ENDIF
          \IF {$m.t > \lc_i$}
          \STATE $\lc_i$ \attr $m.t$; 
          \STATE $\increment_i$ \attr \TRUE; $\equalincr[i]$ \attr \TRUE;
          \FORALL{$k \neq i$}  
          \STATE $\equalincr_i[k]$ \attr $m.\equalincr_i[k]$; \ENDFOR 
          \ELSIF {$m.t = \lc_i$}
          \STATE $\increment_i$ \attr \TRUE; $\equalincr[i]$ \attr \TRUE;
          \FORALL{$k$}  
          \STATE $\equalincr_i[k]$ \attr $\equalincr_i[k] \lor m.\equalincr_i[k]$; 
          \ENDFOR \ENDIF
          \FORALL{$k \neq i$}
          \IF {$m.\ckptv[k] > \ckptv_i[k]$}
          \STATE $\ckptv_i[k]$ \attr $m.\ckptv[k]$; $\taken_i[k]$ \attr $m.\taken[k]$;
          \ELSIF {$m.\ckptv[k] = \ckptv_i[k]$}
          \STATE $\taken_i[k]$ \attr $\taken_i[k] \lor m.\taken[k]$;
          \ENDIF
          \ENDFOR
          \STATE deliver(m); 
        \end{algorithmic}
    }} \\ \hline
  \end{tabular}
  \caption{Lazy-FI protocol (adapted from~\cite{Tsai2007-lazy}) (Part 2)} \label{alg:lazy-fi2}
\end{figure}

Fortunately, a small variation in the $\greater$ vector allows the
implementation of the lazy strategy. The entry $\greater_i[i]$ is set
to false at every checkpoint and set to true when a process set its
$\increment$ flag~\cite{Tsai2007-lazy}. To differentiate the vector
used in Lazy-FI to the one used in FI, we are going to introduce an
equivalent vector with an intuitive meaning: $\equalincr$. Each entry
$\equalincr_i[k]$ is true if to the knowledge of $\proc{i}$ its clock
is equal to the clock of $\proc{k}$ and $\proc{k}$ will increase its
clock before saving the next checkpoint. When $\equalincr_i[k]$ is
false, we have no additional information whether the clock of
$\proc{i}$ is greater or equal to the clock of $\proc{k}$.

Fig.~\ref{fig:lazy-greaterb} shows the  propagation of $\equalincr$
and it is very similar to Fig.~\ref{fig:lazy-greatera}. Once again,
instead  of true  and false  values, we  have used  the signs  $+$ and
$\geq$.   Due to  an extra  message from  $\proc{5}$, $\proc{3}$  will
increase  its clock  before  saving the  next checkpoint.   $\proc{2}$
receives this information and does not take a forced checkpoint before
delivering $m_7$.

Fig.~\ref{fig:lazy-greaterc} shows another similar situation that
emphasizes the need of the vectors $\ckptv$ and $\taken$. Although
upon the reception of $m_7$ $\proc{2}$ receives the information that
$\proc{3}$ will increase its clock, $\proc{2}$ will take a forced
checkpoint to break the Z-cycle $[m_7, m_4, m_6]$.

The conditions used in the Lazy-FI protocol can be stated as follows:

\begin{center}
  $\clazyfi \equiv \clazyfione \lor \clazyfitwo$

  \smallskip
  
  $\clazyfione \equiv \exists k:
       \sentto_i[k] \land \lnot m.\equalincr[k] \land$\\
  $m.t > \lc_i$

  \smallskip
  
  $\clazyfitwo \equiv \cfitwo$

\end{center}

Figs.~\ref{alg:lazy-fi1} and \ref{alg:lazy-fi2} present the code that implements the Lazy-FI
protocol~\cite{Tsai2007-lazy} using the $\equalincr$ vector.

\section{Attempts to optimize FI and Lazy-FI} \label{sec:attempts}

This section starts with a description of the FINE approach to
optimize FI. After that, it presents a property that must be followed
by all optimizations of FI and proves that both FINE and Lazy-FINE do
not guarantee the absence of useless checkpoints.

\subsection{The FINE approach}

The  basic  FINE  protocol  tries  to  reduce  the  number  of  forced
checkpoints  using  the  same  data structures  as  the  FI  protocol.
Fig.~\ref{fig:fine-proposal} illustrates the approach.  $\proc{2}$ has
sent a message  $m_1$ to $\proc{3}$ with $m_1.t =  1$. When $\proc{2}$
receives $m_3$ from  $\proc{1}$, it verifies that $m_3.t  > m_1.t$ and
$\lc_3$, up  to $\proc{2}$'s  knowledge, has not  reached 2  yet.  The
receiving of $m_3$ would have forced  a checkpoint in the FI protocol,
but since $m_3.\taken[2]$ is false, the FINE protocol does not force a
checkpoint because the  messages close no Z-cycle.   These sequence of
messages are called harmless cycles~\cite{Luo2009}.

\begin{figure}[htbp]
  \centering
{\begin{picture}(230,130)(0,-15)

\put(0,85){
  \put(0,0){\large $\proc{1}$}
  \put(20,0){\vector(1,0){210}}
  \put(20,0){\drawckpttabcldv{$\myclock{1}$ & \gmark & \gmark}
            {\nmark{1} & \umark{0} & \umark{0}}}
  \put(60,0){\drawckptijtabcldv{1}{2}{$\myclock{2}$ & \gmark & \gmark}
            {\nmark{2} & \umark{0} & \umark{0}}}
  \put(143,0){\drawtabcldv{$\myclock{2}$ & \gmark & \gmark}
             {\nmark{2} & \nmark{1} & \nmark{1}}}
  \put(178,0){\vector(1,-2){20}}
  \put(196,-10){\drawmitabcldv{3}{$\myclock{2}$ & \gmark & \gmark}
               {\nmark{2} & \nmark{1} & \nmark{1}}}
  }

\put(0,45){
  \put(0,0){\large $\proc{2}$}
  \put(20,0){\vector(1,0){210}}
  \put(20,0){\drawckpttabcldv{\gmark & $\myclock{1}$ & \gmark}
            {\umark{0} & \nmark{1} & \umark{0}}}
  \put(60,0){\vector(1,-2){20}}
  \put(78,-10){\drawmitabcldv{1}{\gmark & $\myclock{1}$ & \gmark}
      {\umark{0} & \nmark{1} & \umark{0}}}
  \put(180,2){\notfckptmark}     
  \put(205,0){\drawtabcldvunder{\emark & $\myclock{2}$ & \gmark}
             {\nmark{2} & \nmark{1} & \nmark{1}}}
  }

\put(0,5){
  \put(0,0){\large $\proc{3}$}
  \put(20,0){\vector(1,0){210}}
  \put(20,0){\drawckpttabcldv{\gmark & \gmark & $\myclock{1}$}
            {\umark{0} & \umark{0} & \nmark{1}}}
  \put(80,0){\drawtabcldvunder{\gmark & \emark & $\myclock{1}$}
            {\umark{0} & \nmark{1} & \nmark{1}}}
  \put(120,0){\vector(1,4){20}}
  \put(132,24){\drawmitabcldv{2}{\gmark & \emark & $\myclock{1}$}
              {\umark{0} & \nmark{1} & \nmark{1}}}
  \put(170,0){\drawckptijtabcldv{3}{2}{\gmark & \gmark & $\myclock{2}$}
            {\umark{0} & \umark{1} & \nmark{2}}}
  }				 

\end{picture}}
  \caption{Fine proposal}
  \label{fig:fine-proposal}
\end{figure}

The basic FINE protocol is based on the following
condition~\cite{Luo2009}:

\begin{center}
  $\cfine \equiv \cfineone \lor \cfinetwo$
\end{center}

where condition $\cfineone$ can be expressed using
a $\greater$ vector~\cite{Luo2011} and $\cfinetwo$
is equivalent to $\cfitwo$.

\begin{center}
$\cfineone \equiv \exists k: \sentto_i[k] \land
m.\greater[k] \land$ \mbox{$m.t > \lc_i \land$}
$m.\taken[k]
$  
\smallskip

$\cfinetwo \equiv \cfitwo$
\end{center}

Fig.~\ref{alg:fine} presents the code  that implements the $\cfineone$
predicate.   The  complete  basic  FINE  protocol  is  implemented  by
replacing    the   \textbf{FI\_1()}    with   \textbf{FINE\_1()}    in
Fig.~\ref{alg:orig-HMNR2}.

\begin{figure}[htbp]
  \centering
 \fbox{\parbox{8.5cm}{\textbf{FINE\_1()}
\begin{algorithmic}
  \RETURN {$\exists k:\sentto_i[k] \land m.\greater[k] \land$}
  \STATE \hspace{1em} $m.t > \lc_i \land\ m.\taken[k]$ 
\end{algorithmic}}}
\caption{Checkpoint-inducing condition for $\cfineone$~\cite{Luo2011}}  \label{alg:fine}
\end{figure}

A lazy version of the FINE protocol, called Lazy-FINE, has been
proposed in the literature~\cite{Luo2011}. Let us define the
checkpoint inducing conditions using the vector $\equalincr$:

\begin{center}
  $\clazyfine \equiv \clazyfineone \lor \clazyfinetwo$

  \smallskip
  
  $\clazyfineone \equiv \exists k: \sentto_i[k] \land \lnot m.\equalincr[k]\ \land$

  $m.t > \lc_i[i] \land m.taken[i]$

  \smallskip

  $\clazyfinetwo \equiv \cfitwo$  
\end{center}

Fig.~\ref{alg:lazyfine} presents the code that implements the
$\clazyfineone$ predicate.  The complete basic Lazy-FINE protocol is
implemented by replacing the \textbf{Lazy\_FI\_1()} with
\textbf{Lazy\_FINE\_1()} in Fig.~\ref{alg:lazy-fi2}.  

\begin{figure}[htbp]
  \centering
 \fbox{\parbox{8.5cm}{\textbf{Lazy-FINE\_1()}
\begin{algorithmic}
  \RETURN {$\exists k:\sentto_i[k] \land \lnot m.\equalincr[k] \land$}
  \STATE \hspace{1em} $m.t > \lc_i \land\ m.\taken[k]$ 
\end{algorithmic}}}
\caption{Checkpoint-inducing condition for $\clazyfineone$ (adapted
  from~\cite{Luo2011})} \label{alg:lazyfine}
\end{figure}

\subsection{FI's optimization limits} \label{sec:limitations}

The timestamps of Fig.~\ref{fig:fine-proposal} are not Z-consistent,
since $\ckpt{1}{2} \zpath \ckpt{3}{2}$ and $\ckpt{1}{2}.t =
\ckpt{3}{2}.t$.  This violation of Z-consistency may appear innocuous
at first, but it violates an important property of any FI
optimization. Suppose one considers the $\cfione \land \pred \lor
\cfitwo$ as capable of producing a more efficient protocol, and
$\cfitwo$ is kept exactly as in the original
condition. Theorem~\ref{theo_zconsistent_timestamping} proves that
$\cfione \land \pred$ must enforce a Z-consistent timestamping to be a
valid optimization.

\begin{theorem} \label{theo_zconsistent_timestamping}
  Any  optimization  of  the  FI  protocol  whose  checkpoint-inducing
  condition can  be expressed as  $\cfione \land \pred  \lor \cfitwo$,
  $\cfione \land \pred$ must enforce a Z-consistent timestamping.
\end{theorem}

\begin{proof}  Assume an optimization of the FI protocol with condition  
 $\cfione \land  \pred \lor \cfitwo$ where $\cfione  \land \pred$ does
  not  enforce a  Z-consistent  timestamping. Thus,  there  must be  a
  computation  with two  checkpoints  $\ckpt{i}{x}$ and  $\ckpt{j}{y}$
  such  as $\ckpt{i}{x}  \zpath \ckpt{j}{y}$  and  $\ckpt{i}{x}.t \geq
  \ckpt{j}{y}.t$.  For  simplicity,  let  $\ckpt{i}{x}.t  =  \alpha  +
  \delta$, $\ckpt{j}{y}.t  = \alpha$ with $\delta \geq  0$ and $\zeta$
  be   the  zigzag  path   between  $\ckpt{i}{x}$   and  $\ckpt{j}{y}$
  (Fig.~\ref{fig:theo_zconsistencya}). Depending  on the properties
  of  $\cfione  \land  \pred$,  this computation  can  be  arbitrarily
  complex,  involving  other  processes  and requiring  many  message
  exchanges. Let  us assume that  this computation does not  enforce a
  Z-consistent timestamping, but has no useless checkpoint.

  We add to the computation another message $m_j$ sent by $\proc{j}$
  as the \emph{first} event of the interval $\interv{j}{y}$ and
  received by $\proc{i}$ in the interval $\interv{i}{x}$ after the
  first message of $\zeta$ is sent
  (Fig.~\ref{fig:theo_zconsistencyb}).  By our construction $m_j.t
  = \alpha$.  This implies that $\ckpt{i}{x} \nrightarrow
  \ckpt{j}{y}$, and that $\cfitwo$ can never be true.  Evaluating
  $\cfione$ at the time $m_j$ is received,
  $m_j.t \leq \alpha + \delta$ and no
  checkpoint is forced upon the reception of $m_j$, no matter the
  existence of $\pred$. Since $\zeta$ must be non-causal, the z-cycle
  formed by $m_j$ and $\zeta$ is not detected by $\cfitwo$ and the protocol 
  allowed the occurrence of a useless checkpoint $\ckpt{j}{y}$. \qedhere
\end{proof}

\begin{figure}[htbp]
  \centering
  \subfloat[]{
{\begin{picture}(120,60)(0,-5)
      \put(0,40){
        \put(0,0){$\proc{i}$}
        \put(20,0){\vector(1,0){100}}
        \put(20,0){\drawckptijts{i}{x}{$\alpha+\delta$}}
        \put(60,0){
	  \put(0,0){\line(1,-2){7.5}}
          \put(2.5,-15){\line(1,0){5}}
	  \put(2.5,-15){\vector(1,-2){7.5}}
          \put(-5,-15){$\zeta$}	
	}
      }

      \put(0,10){
        \put(0,0){$\proc{j}$}
        \put(20,0){\vector(1,0){100}}
        \put(100,0){\drawckptijts{j}{y}{$\alpha$}}
      }
  \end{picture}}
 \label{fig:theo_zconsistencya}}
  
  \subfloat[]{
{\begin{picture}(160,60)(0,-5)
      \put(0,40){
        \put(0,0){$\proc{i}$}
        \put(20,0){\vector(1,0){140}}
        \put(20,0){\drawckptijts{i}{x}{$\alpha+\delta$}}
        \put(60,0){
	  \put(0,0){\line(1,-2){7.5}}
          \put(2.5,-15){\line(1,0){5}}
	  \put(2.5,-15){\vector(1,-2){7.5}}
          \put(-5,-15){$\zeta$}	
	}
      }

      \put(0,10){
        \put(0,0){$\proc{j}$}
        \put(20,0){\vector(1,0){140}}
        \put(90,0){\drawckptijts{j}{y}{$\alpha$}}
        \put(110,0){\vector(1,2){15}}
        \put(120,12){\drawmits{j}{$\alpha$}}	
      }
  \end{picture}}
 \label{fig:theo_zconsistencyb}}
  \caption{Necessity of a Z-consistent timestamping}
  \label{fig:theo_zconsistency}
\end{figure}

\subsubsection*{FINE may lead to useless checkpoints}

FINE is an optimization of  FI where the checkpoint inducing condition
can  be  expressed  as  $\cfione  \land  \pred  \lor  \cfitwo$,  where
$\pred \equiv  m.\taken[k]$, evaluated  by process $\proc{i}$  when it
receives a message $m$. Being an  optimization, we can assume there is
at least one  situation where $\cfione$ is true,  but $m.\taken[k]$ is
false.  In  this case, because  $\cfione$ is true, we  know $\proc{i}$
has sent  a message  $m'$ to  $\proc{k}$ in  the current  interval and
$m'.t  <  m.t$.   We  also known  that,  considering  all  information
available to $\proc{i}$, including what it learns trough the reception
of $m$, it can't be sure $lc_k  \geq m.t$ after the reception of $m'$.
Process $\proc{i}$ can only  conservatively assume that $\proc{k}$ has
updated its clock to $m'.t$.  Thus, FI forces a checkpoint to ensure a
Z-consistent timestamping.

Even as $\cfione$ is true, $\pred$ is  false and FINE does not force a
checkpoint.  However,  $m.\taken[k]$ only informs $\proc{i}$  that the
causal  paths  from  the  last checkpoint  taken  by  $\proc{k}$  that
causally precedes the reception of $m$ to the current interval contain
no  checkpoints  from  other   processes.   There  is  no  information
available  in this  data  structure to  exclude  the possibility  that
$\proc{k}$ has indeed updated $lc_k$ to exactly $m'.t$.  If this turns
out  to  be  the  case, the  checkpoint  $\ckpt{k}{x}$  following  the
reception of $m'$ by $\proc{k}$ can  be taken before $lc_k$ is further
incremented,  and as  a consequence  $\ckpt{k}{x}.lc \leq  m.t$.  This
opens the possibility of violating Z-consistent timestamping.  This is
exactly   what   happens   when   message   $m_3$   is   received   in
Fig.~\ref{fig:fine-proposal},  that   illustrates  a   situation  FINE
doesn't force a checkpoint FI would force.

According to Theorem~\ref{theo_zconsistent_timestamping}, because 
$\cfineone$ doesn't  produce a  Z-consistent timestamping it  does not
guarantee the absence of useless checkpoints. 
Indeed, Fig.~\ref{fig:counterexample} is a counterexample: it shows a
possible continuation of the scenario presented in
Fig.~\ref{fig:fine-proposal} that leads to the occurrence of a useless
checkpoint.  When $\proc{2}$ receives $m_3$ from $\proc{1}$, there is
no Z-cycle \emph{known} to $\proc{2}$ closed by the receipt of
$m_3$. However, this does not exclude the formation of a Z-cycle,
undetected at the time $m_3$ is received.

\begin{figure}[htbp]
  \centering 
{\begin{picture}(230,130)(0,-15)

\put(0,85){
  \put(0,0){\large $\proc{1}$}
  \put(20,0){\vector(1,0){210}}
  \put(20,0){\drawckpttabcldv{$\myclock{1}$ & \gmark & \gmark}
            {\nmark{1} & \umark{0} & \umark{0}}}
  \put(60,0){\drawckptijtabcldv{1}{2}{$\myclock{2}$ & \gmark & \gmark}
            {\nmark{2} & \umark{0} & \umark{0}}}
  \put(102,0){\drawtabcldv{$\myclock{2}$ & \gmark & \gmark}
             {\nmark{2} & \nmark{1} & \nmark{1}}}
  \put(138,0){\vector(1,-2){20}}
  \put(156,-10){\drawmitabcldv{3}{$\myclock{2}$ & \gmark & \gmark}
               {\nmark{2} & \nmark{1} & \nmark{1}}} 
  \put(200,0){\drawtabcldv{$\myclock{2}$ & \gmark & \emark}
             {\nmark{2} & \umark{1} & \nmark{2}}}
}

\put(0,45){
  \put(0,0){\large $\proc{2}$}
  \put(20,0){\vector(1,0){210}}
  \put(20,0){\drawckpttabcldv{\gmark & $\myclock{1}$ & \gmark}
             {\umark{0} & \nmark{1} & \umark{0}}}
  \put(40,0){\vector(1,-2){20}}
  \put(56,-10){\drawmitabcldv{1}{\gmark & $\myclock{1}$ & \gmark} 
    {\umark{0} & \nmark{1} & \umark{0}}}
  \put(155,0){\drawtabcldvunder{\emark & $\myclock{2}$ & \gmark}
             {\nmark{2} & \nmark{1} & \nmark{1}}}
}

\put(0,5){
  \put(0,0){\large $\proc{3}$}
  \put(20,0){\vector(1,0){210}}
  \put(20,0){\drawckpttabcldv{\gmark & \gmark & $\myclock{1}$}
            {\umark{0} & \umark{0} & \nmark{1}}}
  \put(60,0){\drawtabcldvunder{\gmark & \emark & $\myclock{1}$}
            {\umark{0} & \nmark{1} & \nmark{1}}}
  \put(82,0){\vector(1,4){20}}
  \put(104,70){\drawmitabcldv{2}{\gmark & \emark & $\myclock{1}$}
    {\umark{0} & \nmark{1} & \nmark{1}}}
  \put(110,0){\drawckptijtabcldv{3}{2}{\gmark & \gmark & $\myclock{2}$}
       {\umark{0} & \umark{1} & \nmark{2}}}
  \put(180,0){\vector(1,4){20}}
  \put(190,30){\drawmitabcldv{4}{\gmark & \gmark & $\myclock{2}$}
    {\umark{0} & \umark{1} & \nmark{2}}}
}				 

\end{picture}} 
  \caption{A useless checkpoint under FINE}
  \label{fig:counterexample}
\end{figure}

\subsubsection*{Lazy-FINE may lead to useless checkpoints} \label{sec:lazy-fine}

Lazy-FINE is an optimization of  Lazy-FI where the checkpoint inducing
condition        can        also         be        expressed        as
$\clazyfione     \land     \pred      \lor     \clazyfitwo$,     where
$\pred \equiv  m.\taken[k]$, evaluated  by process $\proc{i}$  when it
receives a  message $m$. By the  same argument we have  made for FINE,
there is at  least one situation where $\clazyfione$  is true, $\pred$
is false  and Lazy-FINE  does not  force a  checkpoint.  This  means a
message $m'$ was  sent to $\proc{k}$ by $\proc{i}$ with  $m'.t < m.t$,
and that checkpoint  $\ckpt{k}{x}$ following the reception  of $m'$ by
$\proc{k}$ can  be taken before  $lc_k$ is further incremented.   As a
consequence  $\ckpt{k}{x}.lc \leq  m.t$, and  we have  once again  the
possibility violating Z-consistent timestamping.

Theorem~\ref{theo_zconsistent_timestamping}  also   informs  us  that,
because the  $\clazyfineone$ predicate doesn't produce  a Z-consistent
timestamping,   it  does   not  guarantee   the  absence   of  useless
checkpoints.   In  Fig.~\ref{fig:lazycounterexample}, when  $\proc{3}$
receives   $m_4$  from   $\proc{2}$,  $m_4.t   >  \lc_3$,   but  since
$\lnot  m.taken[3]$, $[m_4,  m_2,  m_3]$ would  form  just a  harmless
Z-cycle.  However,  $\proc{3}$  receives $m_5$  from  $\proc{4}$  with
$m_5.t  =  \lc_3$  no  forced   checkpoint  is  taken  and  a  Z-cycle
$m_5, m_4, m_2$ is formed.

\begin{figure}[htbp]
  \centering 

\begin{picture}(230,170)(0,-15)

  \put(0,125){
    \put(0,0){\large $\proc{1}$}
    \put(20,0){\vector(1,0){210}}
    \put(20,0){\drawckpttabcldv
      {$\myclock{1}$ & \geqmark & \geqmark & \geqmark} 
      {\nmark{1} & \umark{0} & \umark{0} & \umark{0}}}        
    \put(40,0){\vector(1,-4){10}}
    \put(50,-10){\drawmitabcldv{1}
      {$\myclock{1}$ & \geqmark & \geqmark & \geqmark}
      {\nmark{1} & \umark{0} & \umark{0} & \umark{0}}}
  }
  
  \put(0,85){
    \put(0,0){\large $\proc{2}$}
    \put(20,0){\vector(1,0){210}}
    \put(20,0){\drawckpttabcldv
      {\geqmark & $\myclock{1}$ & \geqmark & \geqmark} 
      {\umark{0} & \nmark{1} & \umark{0} & \umark{0}}}
    \put(46,-8){$+$}      
    \put(80,0){\drawckpttabcldv
      {\geqmark & $\myclock{2}$ & \geqmark & \geqmark} 
      {\umark{1} & \nmark{2} & \umark{0} & \umark{0}}}
    \put(130,0){\vector(1,-4){10}}
    \put(140,-10){\drawmitabcldv{4}
      {\geqmark & $\myclock{2}$ & \geqmark & \geqmark} 
      {\umark{1} & \nmark{2} & \nmark{1} & \nmark{1}}}
    \put(186,4){$+$}      
  }

  \put(0,45){
    \put(0,0){\large $\proc{3}$}
    \put(20,0){\vector(1,0){210}}
    \put(20,0){\drawckpttabcldv
      {\geqmark & \geqmark & $\myclock{1}$ & \geqmark} 
      {\umark{0} & \umark{0} & \nmark{1} & \umark{0}}}
    \put(40,0){\vector(1,-4){10}}
    \put(50,-10){\drawmitabcldv{2}
      {\geqmark & \geqmark & $\myclock{1}$ & \geqmark} 
      {\umark{0} & \umark{0} & \nmark{1} & \umark{0}}}
    \put(120,2){\notfckptmark}
    \put(136,-8){$+$}	
  }

  \put(0,5){
    \put(0,0){\large $\proc{4}$}
    \put(20,0){\vector(1,0){210}}
    \put(20,0){\drawckpttabcldv
      {\geqmark & \geqmark & \geqmark & $\myclock{1}$} 
      {\umark{0} & \umark{0} & \umark{0} & \nmark{1}}}
    \put(46,-8){$+$}
     \put(89,0){\vector(1,4){20}}
     \put(101,30){\drawmitabcldv{3}
      {\geqmark & \geqmark & \geqmark & $\tsincr{\myclock{1}}$}
      {\umark{0} & \umark{0} & \nmark{1} & \nmark{1}}}
    \put(140,0){\drawckptijtabcldv{3}{2}
      {\geqmark & \geqmark & \geqmark & $\myclock{2}$} 
      {\umark{0} & \umark{0} & \umark{1} & \nmark{2}}}
    \put(170,0){\vector(1,4){20}}
    \put(182,30){\drawmitabcldv{5}
      {\geqmark & \geqmark & \geqmark & $\myclock{2}$} 
      {\umark{0} & \umark{0} & \umark{1} & \nmark{2}}}
  }				 

\end{picture}
  \caption{A useless checkpoint under Lazy-FINE}
  \label{fig:lazycounterexample}
\end{figure}

\section{Conclusion} \label{sec:conclusion}

This paper reviewed index-based checkpointing protocols and proved
that the FINE and Lazy-FINE protocols do not guarantee the absence of
useless checkpoints.  This paper also reinforces that all
optimizations of FI must enforce a Z-consistent timestamping. As a
consequence, FI and Lazy-FI can be rolled back to the position of most
efficient index-based protocols; whether or not they can be further
optimized remains an open problem.

\newpage

\bibliographystyle{elsarticle-num}
\bibliography{rollback}   

\end{document}